\documentclass{IEEEtran}

\usepackage[utf8]{inputenc}
\usepackage[english]{babel}
\usepackage[T1]{fontenc}
\usepackage{amsmath,amsthm,commath}  
\usepackage{physics}
\usepackage{tikz}
\usepackage{lipsum}
\usepackage{mathtools}
\usepackage{braket}
\usepackage{graphicx}
\usepackage{amssymb,dsfont,ctable}

\usepackage{hyperref}
\definecolor{darkred}  {rgb}{0.5,0,0}
\definecolor{darkblue} {rgb}{0,0,0.5}
\definecolor{darkgreen}{rgb}{0,0.5,0}
\hypersetup{
	colorlinks = true,
	urlcolor  = blue,         
	linkcolor = darkblue,     
	citecolor = darkgreen,    
	filecolor = darkred       
}

\input{qcircuit}

\newcommand{\lin}[1]{\mathcal{L}(#1)}
\newcommand{\hilb}[1]{\mathcal{#1}}

\newcommand{\id}{\operatorname{id}}
\newcommand{\pp}{\operatorname{Pos}}
\newcommand{\diag}{\operatorname{diag}}
\newcommand{\ch}[1]{\mathcal{#1}}

\newcommand{\openone}{\mathds{1}}

\newtheorem{dfn}{Definition}
\newtheorem{lmm}{Lemma}

\newtheorem{prop}{Proposition}

\begin{document}
\title{Explicit  construction   of  optimal   witnesses  for
  input--output correlations attainable by quantum channels}

\date{\today}

\author{
  Michele  Dall'Arno,  Sarah Brandsen,  and  Francesco
  Buscemi

  \thanks{M.~D.  acknowledges support from MEXT Quantum Leap
    Flagship      Program      (MEXT      Q-LEAP)      Grant
    No.   JPMXS0118067285,   JSPS   KAKENHI   Grant   Number
    JP20K03774,  and  the  International  Research  Unit  of
    Quantum    Information,    Kyoto   University.     F.~B.
    acknowledges  support from  MEXT  Quantum Leap  Flagship
    Program (MEXT Q-LEAP), Grant Number JPMXS0120319794, and
    from  the Japan  Society  for the  Promotion of  Science
    (JSPS)  KAKENHI,  Grants   Number  19H04066  and  Number
    20K03746.}

  \thanks{M.   Dall'Arno is  with the  Yukawa Institute  for
    Theoretical  Physics,  Kyoto  University,  Kitashirakawa
    Oiwakecho, Sakyoku, Kyoto 606-8502,  Japan, and with the
    Faculty of  Education and Integrated Arts  and Sciences,
    Waseda University, 1-6-1 Nishiwaseda, Shinjuku-ku, Tokyo
    169-8050, Japan (dallarno.michele@yukawa.kyoto-u.ac.jp)}

  \thanks{S.  Brandsen  is with  the Department  of Physics,
    Duke  University,  Durham,  North  Carolina  27708,  USA
    (sarah.brandsen@duke.edu)}

  \thanks{F.   Buscemi  is  with   the  Graduate  School  of
    Informatics,  Nagoya   University,  Chikusa-ku,  Nagoya,
    464-8601, Japan (buscemi@i.nagoya-u.ac.jp)}

  \thanks{Report number: YITP-20-110}
}

\maketitle

\begin{abstract}
  Given a  quantum channel---that is, a  completely positive
  trace-preserving  linear map---as  the only  communication
  resource available  between two  parties, we  consider the
  problem  of  characterizing  the set  of  classical  noisy
  channels that can be obtained from it by means of suitable
  classical-quantum    encodings    and    quantum-classical
  decodings,   respectively,  on   the   sender's  and   the
  receiver's side.   We consider  various classes  of linear
  witnesses and compute their  optimum values in closed form
  for several  classes of  quantum channels.   The witnesses
  that    we    consider     here    are    formulated    as
  \textit{communication games},  in which Alice's aim  is to
  exploit a  single use of  a given quantum channel  to help
  Bob  guess  some  information  she has  received  from  an
  external referee.
\end{abstract}

\section{Introduction}

Suppose Alice   and   Bob  play   a
two-party, quantum--enhanced version of the  popular game charades.  In
each run,  Alice's aim is  to help  Bob guess some  piece of
information that she has received from a referee.  As in the
traditional  charades game,  there  is a  bottleneck in  the
communication channel, in this case created by a given noisy
quantum channel.   After each round, the  referee provides a
payoff  that  depends  on  both the  information  Alice  was
provided and Bob's  guess.  The parties' aim  is to maximize
the average  payoff, which depends only on the  channel and
the game, by optimizing Alice's encoding and Bob's decoding.

Here, we  introduce the  communication utility of  any given
quantum  channel for  any  given communication  game as  the
average payoff after asymptotically many runs. Communication
games are  linear functionals (i.e.,  \textit{witnesses}) on
the set  of classical noisy channels  (i.e., \textit{quantum
  signaling  correlations}) that  can be  obtained from  the
given quantum  channel, and the  corresponding communication
utility  constitutes  the  optimal  value  for  any  such  a
witness.

For any given channel and game, the problem of computing the
communication  utility,  as  well as  the  encoding-decoding
achieving  it, can  be generally  framed as  a semi-definite
programming  problem.  However,  here we  are interested  in
those cases  where a  closed-form solution is  possible.  To
this aim, we restrict to the following classes of games:
\begin{itemize}
\item  {\em Unbiased  games},  where any  of Bob's  possible
  outcomes generates the same average payoff;
\item  {\em  Discrimination games}  where  the  payoff is  a
  diagonal matrix;
\item  {\em   Binary-output  games}  where   Bob has two  possible
  outcomes;
\item  {\em Binary-input-output  games}  where Alice has two possible inputs and Bob has two possible outcomes.
\end{itemize}

For any arbitrary game,  we derive the communication utility
of  any   unitary,  trace-class,  erasure,   dephasing,  and
quantum-classical channel, generalizing  a result by Frenkel
and Weiner~\cite{FW13} that applies to the identity channel.
For any  unbiased game, we derive  the communication utility
of  any  depolarizing channel.   We  apply  our result  for
unitary  channels  to  the  case  of  discrimination  games,
providing  a  simplified proof  of  a  result by  Elron  and
Eldar~\cite{EE07}. We  show that  any binary-output  game is
either  trivial,  or  equivalent  to  a  binary-input-output
discrimination game. For  any such game, we  show that the
optimal  encoding consists  of two  orthogonal pure  states,
regardless  of  whether  the  channel  is  commutativity
preserving.    Using  these   facts,   we  extend   previous
results~\cite{DBB17,  BD19}  to   derive  the  communication
utility  of  any  Pauli, amplitude-damping,  optimal  1-to-2
universal cloning, and  shifted-depolarizing channel for any
binary-output game.

The     paper    is     structured    as     follows.     In
Section~\ref{sec:utility}  we   formalize  the   problem  of
evaluating  the communication  utility of  quantum channels.
We   address  such   a  problem   for  arbitrary   games  in
Section~\ref{sec:arbitrary}.  We  specialize our  results to
the  cases of  unbiased,  discrimination, and  binary-output
games             in            Sections~\ref{sec:unbiased},
\ref{sec:discrimination},        and       \ref{sec:binary},
respectively.  Finally, we summarize our results and discuss
some outlooks in Section~\ref{sec:conclusion}.

\section{Utility of quantum channels}
\label{sec:utility}

We recall some standard  facts from quantum information
theory~\cite{Wil13}.  Any quantum  system is associated
with  a Hilbert  space $\hilb{H}$,  and we  denote with
$\lin{\hilb{H}}$  the  space  of  linear  operators  on
$\hilb{H}$.   Any   quantum  state  in   $\hilb{H}$  is
represented       by       a       density       matrix
$\rho   \in   \lin{\hilb{H}}$,    namely   a   positive
semidefinite   unit-trace    operator.    Any   quantum
transformation   from  $\hilb{H}$   to  $\hilb{K}$   is
represented       by       a      quantum       channel
$\ch{C}: \lin{\hilb{H}}  \to \lin{\hilb{K}}$,  namely a
completely-positive  trace-preserving linear  map.  Any
quantum measurement  on $\hilb{H}$ is represented  by a
positive-operator      valued       measure      (POVM)
$\{ \pi_y  \in \lin{\hilb{H}}  \}$, namely a  family of
positive-semidefinite      operators     such      that
$\sum_y \pi_y = \openone$, where $\openone$ denotes the
identity operator.

The  communication  setup  under consideration  can  be
framed as a  quantum game played by  two parties, Alice
and Bob,  and an  external referee.  Prior  to starting
the game, the parties are allowed to establish a common
strategy.  In each run, the following occurs:
\begin{enumerate}
\item The  referee gives  as an input  to Alice  the value
  $x \in  [1, n]$ of  a random variable $X$  with prior
  probability $p_x$,  known in advance to  the parties;
\item  Alice and  Bob  are allowed  to perform  one-way
  communication over a single  use of a quantum channel
  $\ch{C}$ from  Alice to  Bob;
\item   Bob   returns   to  the   referee   the   value
  $y \in [1, m]$ of a random variable $Y$;
\item  The  referee  provides  a  payoff  according  to
  the function $u_{x,y}\in\mathbb{R}$, known in advance to the parties.
\end{enumerate}
  
In the  absence of any previously  shared resource, the
most general strategy allowed  by quantum theory is for
Alice  to encode  her input  $x$ into  a quantum  state
$\rho_x$,   and   for   Bob   to   decode   his   input
$\ch{C}(\rho_x)$ by means of a POVM $\{ \pi_y \}$.
The resulting setup is as follows
\vspace{0.3cm}
\begin{align*}
  \begin{aligned}
    \Qcircuit   @C=12pt  @R=10pt   {  &   \lstick{x}  &
      \ustick{[1,n]}   \cw   &   \cprepareC{\rho_x}   &
      \ustick{\hilb{H}}  \qw &  \gate{\ch{C}} &
      \ustick{\hilb{K}}  \qw  &  \measureD{\pi_y}  &
      \ustick{[1,m]} \cw & \rstick{y} \cw }
  \end{aligned}
\end{align*}
Then the expected average payoff is given by
\begin{align}
  \label{eq:avgpayoff}
  \langle u \rangle_{\{\rho_x\},  \{ \pi_y \}}  =   \sum_{x=1}^n  \sum_{y=1}^m   p_x
  \Tr[\ch{C}(\rho_x) \pi_y] u_{x,y}.
\end{align}
From  Eq.~\eqref{eq:avgpayoff}  it immediately  follows
that    two    games    with    the    same    quantity
\begin{align*}
  g_{x,y} := p_x u_{x,y}
\end{align*}
have the same  average payoff, hence $g$  identifies a class
of equivalence among games.  We  can now introduce a measure
of  how  useful a  channel  $\ch{C}$  is in  maximizing  the
average payoff of a game $g$.

\begin{dfn}[Communication utility]
  \label{def:utility}
  The {\em  communication utility}  $U(\ch{C}, g)$
  of a  quantum channel  $\ch{C}$ for game  $g$ is
  the  maximum over  any  encoding $\{  \rho_x \}$  and
  decoding  $\{   \pi_y  \}$  of  the   average  payoff
  $\langle u \rangle_{ \{\rho_x\}, \{\pi_y \} }$, namely
  \begin{align*}
    U(\ch{C}, g) =  \max_{\{\rho_x\}, \{\pi_y\}} \langle u \rangle_{ \{\rho_x\}, \{\pi_y \} }.
  \end{align*}
\end{dfn}

In the remaining of this  work we derive the utility of
several classes  of quantum  channels. For  clarity, we
provide a  glossary of the channels  considered in this
work  in  Tab.~\ref{tab:channels}.  We also  provide  a
summary      of      our      main      results      in
Tab.~\ref{tab:utilities}.

\begin{table}[h]
  \begin{tabular}{| >{$}m{.6\columnwidth}<{$} | >{$}m{.3\columnwidth}<{$} |}
    \hline
    {\textrm{\bf Definition}} &  {\textrm{\bf Name}} \\
    \hline
    \hline
    \ch{U}(\rho) := U \rho U^\dagger & \textrm{Unitary} \\
    \hline
    \ch{F}_\lambda(\rho) := \lambda \rho + (1-\lambda) \sum_k \bra{k} \rho \ket{k} \ket{k} \! \! \bra{k} & \textrm{Dephasing} \\
    \hline
    \ch{T}(\rho) := \Tr[\rho] \sigma & \textrm{Trace-class} \\
    \hline
    \ch{E}_\lambda(\rho) := \lambda \rho \oplus (1-\lambda) \Tr[\rho] \ket{\phi} \! \! \bra{\phi} & \textrm{Erasure} \\
    \hline
    \ch{M}(\rho) := \sum_y \Tr[ \rho \pi_y] \ket{y} \! \! \bra{y} & \textrm{QC} \\
    \hline
    \ch{D}_\lambda(\rho) := \lambda \rho + (1-\lambda) \Tr[\rho] \openone/d & \textrm{Depolarizinig} \\
    \hline
    \ch{P}_{\vec\lambda} (\rho) := \lambda_0 \rho + \sum_{k=1}^3 \lambda_k \sigma_k \rho \sigma_k & \textrm{Pauli} \\
    \hline
    \ch{A}_\eta \begin{pmatrix}
      1-\beta & \gamma \\
      \gamma^* & \beta
    \end{pmatrix} = \begin{pmatrix}
      1-\eta \beta & \sqrt{\eta}\gamma \\
      \sqrt{\eta}\gamma^* & \eta \beta
    \end{pmatrix} & \textrm{Amplitude damping} \\
    \hline
    \ch{S}_\lambda (\rho) := \lambda \rho + (1-\lambda)
    \Tr[\rho] \sigma & \textrm{Shifted depolarizing} \\
    \hline
    \ch{N}(\rho) := \frac2{d + 1} P_s (\rho \otimes
    \openone) P_s & \textrm{$1 \to 2$ Cloning} \\
    \hline
  \end{tabular}
  \caption{Glossary of  quantum channels  considered in
    this work.   All channels  are formally  defined in
    the  text.  Not  otherwise  specified channels  are
    denoted with $\ch{C}$ in the text.}
  \label{tab:channels}
\end{table}

\begin{table}[h]
  \begin{tabular}{| >{$}m{.05\columnwidth}<{$} | >{$}m{.25\columnwidth}<{$} | >{$}m{.6\columnwidth}<{$} |}
    \hline
    {\ch{C}} & {\textrm{\bf Game } g} & {\textrm{\bf Utility } U(\ch{C}, g)} \\
    \hline
    \hline
    {\bf \ch{F}_\lambda} & \textrm{Any} & U(\ch{U}, g) \\
    \hline
    {\bf \ch{T}} & \textrm{Any} & \max_y \sum_x  g_{x,y} \\
    \hline
    {\bf \ch{E}_\lambda} & \textrm{Any} & \lambda U(\ch{U}, g) + (1-\lambda) \max_{y} \sum_x g_{x,y} \\ 
    \hline
    {\bf \ch{M}} & \textrm{Any} & \sum_x \lambda_x \\
    \hline
    {\bf \ch{D}_\lambda} & \textrm{Unbiased} &  \lambda U(\ch{U}, g) + (1-\lambda) \sum_x g_{x,0} \\
    \hline
    {\bf \ch{U}} & \textrm{Discrimination} & \sum_{x=0}^{d-1} g_x \Theta(g_x ) \\
    \hline
    {\bf \ch{P}_{\vec\lambda}} & \textrm{Binary} & \max    \left(g_0, \frac{1+\max_{k \ge 1}|2(\lambda_0+\lambda_k)-1|}2 \right)\\
    \hline
    {\bf \ch{A}_\eta} & \textrm{Binary} & \frac{1 + \sqrt{1 - 4p(1-\eta) +4p^2 (1-\eta)}}2 \\
    \hline
    {\bf \ch{S}_\lambda} & \textrm{Binary} & \max\left[g_0, \frac{1 + \lambda + (1-\lambda)(1-2s_{d-1})(2g_0-1)}2 \right] \\
    \hline
    {\bf \ch{N}} & \textrm{Binary} & \frac{d+g_0}{d+1} \\
    \hline
  \end{tabular}
  \caption{Summary  of  our  main results,  namely  the
    utilities $U(\ch{C},  g)$ of  several channels
    $\ch{C}$ for several classes of games.}
  \label{tab:utilities}
\end{table}

\section{Utility for arbitrary games}
\label{sec:arbitrary}

In  this Section  we consider  arbitrary games.  Let us
begin with  some general  results. The  following Lemma
provides a  simple, channel-independent upper  bound to
the utility for any given game.

\begin{lmm}
  \label{thm:bound}
  For any  channel $\ch{C}$ and any  game $g$, the
  utility $U(\ch{C}, g)$ is upper bounded by
  \begin{align*}
    U(\ch{C}, g) \le \sum_x \max_y g_{x,y}.
  \end{align*}
\end{lmm}

\begin{proof}
  By    replacing     the    conditional    probability
  $\Tr[\ch{C}(\rho_x)  \pi_y]$  with an  arbitrary
  conditional  probability  $p_{y|x}$  and  taking  the
  supremum  over  such conditional  probabilities,  one
  clearly has
  \begin{align*}
    U(\ch{C},  g)  \le  \sup_{p_{y|x}}  \sum_{x=1}^n  \sum_{y=1}^m
    p_{y|x} g_{x,y}.
  \end{align*}
  For   any   fixed   $x$,  the   optimal   probability
  distribution      $p_{y|x}$      is     given      by
  $p_{y|x}        =       \delta_{y,y^*}$,        where
  $y^* := \sup_y g_{x,y}$, therefore one has
  \begin{align*}
    \sup_{p_{y|x}} \sum_{x=1}^n  \sum_{y=1}^m p_{y|x}  g_{x,y} =
    \sum_x \sup_y g_{x,y},
  \end{align*}
  from which the statement immediately follows.
\end{proof}

The  following Lemma  characterizes a  subclass of  the
linear  transformations of  game  $g$  under which  the
utility  $U(\ch{C}, g)$  transforms linearly,  for
any channel $\ch{C}$.

\begin{lmm}
  \label{lmm:quasilinearity}
  For any  game $g$  consider the  game $g'$  such that
  $g'_{x,y}  := \alpha  (g_{x,y} +  \beta_x)$, for  any
  $\alpha \ge  0$ and some  $\{ \beta_x \}$.   Then for
  any channel $\ch{C}$ we have
  \begin{align*}
    U(\ch{C}, g') = \alpha \left[ U(\ch{C}, g) + \sum_x \beta_x \right].
  \end{align*}
  Moreover $U(\ch{C}, g)$ and $U(\ch{C}, g')$
  are attained by the same  encoding $\{ \rho_x \}$ and
  decoding $\{ \pi_y \}$.
\end{lmm}

\begin{proof}
  By definition~\ref{def:utility} one immediately has
  \begin{align*}
    U(\ch{C}, g')  := \alpha \sup_{ \{\rho_x\}, \{\pi_y\}}  \sum_{x,y} 
    \Tr[\ch{C}(\rho_x) \pi_y] (g_{x,y} + \beta_x).
  \end{align*}
  Since  POVMs  are  decompositions  of  the  identity,
  namely $\sum_y  \pi_y =  \openone$, and  channels are
  trace-preserving,                              namely
  $\Tr[\ch{C}(\rho_x)] = \Tr[\rho_x] = 1$, one has
  \begin{align*}
    U(\ch{C},  g')  =  \alpha  \left[  \sup_{\rho_x,  \pi_y}
      \sum_{x,y} \Tr[\ch{C}(\rho_x) \pi_y]  g_{x,y} + \sum_x
      \beta_x \right],
  \end{align*}
  for   any   encoding   $\{  \rho_x   \}$   and   POVM
  $\{ \pi_y \}$.  Since only  the first term depends on
  the   encoding   $\{    \rho_x   \}$   and   decoding
  $\{  \pi_y  \}$,  one  has that  $U(\ch{C},  g)$  and
  $U(\ch{C}, g')$ are attained by the same encoding and
  decoding.
\end{proof}

As  an immediate  consequence  of a  recent breakthrough  by
Frenkel and Weiner~\cite{FW13}, the  utility of any identity
quantum channel,  that we denote  by $\id$, is equal  to the
utility  of  any  identity  classical channel  in  the  same
dimension.   Accordingly, the  utility  of  any unitary  and
dephasing  channel  is also  equal  to  the utility  of  the
identity  classical  channel  in   the  same  dimension,  as
follows.

\begin{dfn}[Unitary channel]
  The     action     of     any     unitary     channel
  $\ch{U} : \lin{\hilb{H}}  \to \lin{\hilb{H}}$ on
  any  state  $\rho  \in \lin{\hilb{H}}$  is  given  by
  $\ch{U}(\rho)  =  U  \rho U^\dagger$,  for  some
  unitary $U$.
\end{dfn}

\begin{prop}[Frenkel and Weiner~\cite{FW13}]
  \label{thm:frenkel}
  The  utility  $U(\ch{U},   g)$  of  any  unitary
  channel $\ch{U}$ for any game $g$ is attained by
  an orthonormal encoding $\{\rho_x\}$ and a projective
  decoding  $\{  \pi_y   \}$  that  are  simultaneously
  diagonalizable.
\end{prop}

\begin{dfn}
  The     action    of     the    dephasing     channel
  $\ch{F}_\lambda:        \lin{\hilb{H}}       \to
  \lin{\hilb{H}}$         on          any         state
  $\rho \in \lin{\hilb{H}}$ is given by
  \begin{align*}
    \ch{F}_\lambda(\rho) := \lambda \rho +(1-\lambda) \sum_{k=1}^d
    \bra{k}\rho\ket{k} \ket{k} \! \! \bra{k},
  \end{align*}
  where $d := \dim\hilb{H}$ is the dimension of Hilbert
  space $\hilb{H}$ and $\{ \ket{k} \}$ is some o.n.b.
\end{dfn}

\begin{prop}[Frenkel and Weiner~\cite{FW13}]
  \label{thm:dephasing-utility}
  The  utility   $U(\ch{F}_\lambda,  g)$   of  any
  dephasing channel $\ch{F}_\lambda$  for any game
  $g$   is   attained   by  an   orthonormal   encoding
  $\{\rho_x\}$ and a projective  decoding $\{ \pi_y \}$
  along basis $\{ \ket{k} \}$.
\end{prop}

At the other side of the spectrum of channels there lie
the trace-class  channels, that  is those  channel that
cannot  convey any  information.  Hence,  their utility
corresponds  to a  trivial guessing  on Bob's  side, as
follows.

\begin{dfn}[Trace-type channel]
  \label{dfn:trace-type}
  A channel $\ch{T}$ is  trace-type if and only if
  there   exists    a   state   $\sigma$    such   that
  $\ch{T}(\rho) = \Tr[\rho]  \sigma$ for any state
  $\rho$.
\end{dfn}

\begin{prop}
  For any game $g$,  the utility $U(\ch{T}, g)$ of
  any trace-type channel $\ch{T}$ is given by
  \begin{align*}
    U(\ch{T}, g) = \max_y \sum_x g_{x,y}.
  \end{align*}
  Any encoding is optimal,  and the optimal decoding is
  $\pi_y    =     \delta_{y,y^*}    \openone$,    where
  $y^* := \arg \max_y \sum_x g_{x,y}$.
\end{prop}

\begin{proof}
  The      statement     directly      follows     from
  Definition~\ref{dfn:trace-type}.
\end{proof}

Between  unitary  and  trace-class channels  are
erasure    channels,    that     is,    channels    that
probabilistically  declare  an  error  while  otherwise
achieving noiseless  communication.  Accordingly, their
utility is the  convex combination of the  utility of a
noiseless  channel   and  a  trace-class   channel,  as
follows.

\begin{dfn}[Erasure channel]
  The     action     of     the     erasure     channel
  $\ch{E}_\lambda(\rho)   :   \lin{\hilb{H}}   \to
  \lin{\hilb{H}   \oplus   \hilb{K}}$  on   any   state
  $\rho \in \lin{\hilb{H}}$ is given by
  \begin{align*}
    \ch{E}_\lambda(\rho)  :=  \lambda  \rho  \oplus  (1-\lambda)
    \Tr[\rho] \ket{\phi} \! \! \bra{\phi},
  \end{align*}
  where $\ket{\phi} \in \hilb{K}$.
\end{dfn}

\begin{prop}
  For      any      game     $g$,      the      utility
  $U(\ch{E}_\lambda,   g)$   of  erasure   channel
  $\ch{E}_\lambda$ is given by
  \begin{align*}
    U(\ch{E}_\lambda, g)  = \lambda U(\id, g)  + (1-\lambda)
    U(\ch{T}, g).
  \end{align*}
  By denoting with $\{ \rho_x^*  \}$ and $\{ \pi_y^* \}$ any
  encoding and decoding attaining  $U(\id, g)$, the encoding
  $\{  \rho_x^*  \}$ and  the  decoding  $\{ \pi_y^*  \oplus
  \delta_{y,y^*}     \openone_{\hilb{K}}      \}$     attain
  $U(\ch{E}_\lambda,  g)$, where  $y^* :=  \arg\max_y \sum_x
  g_{x,y}$.
\end{prop}

\begin{proof}
  By direct computation one has
  \begin{align*}
    & U(\ch{E}_\lambda, g) \\
    = & \max_{\{\rho_x\}, \{\pi_y\}} \sum_{x,y} \Big[ \lambda \Tr[\rho_x
        \pi_y] + (1-\lambda) \bra{\phi} \pi_y \ket{\phi} \Big] g_{x,y} \\
    \le  & \lambda  \max_{\{\rho_x\},  \{\pi_y\}} \sum_{x,y}  \Tr[\rho_x
           \pi_y]   g_{x,y}  \\ & +   (1-\lambda)  \max_{\{\pi_y\}}   \sum_{x,y}
                                  \bra{\phi} \pi_y \ket{\phi} g_{x,y},
  \end{align*}
  where     the    maxima     are    over     encodings
  $\{  \rho_x  \in  \lin{\hilb{H}}  \}$  and  decodings
  $\{  \pi_y \in  \lin{\hilb{H}  \oplus \hilb{K}}  \}$.
  Since
  $\Tr[\rho_x  \pi_y]  =  \Tr[\rho_x  P_\hilb{H}  \pi_y
  P_{\hilb{H}}]$                                    and
  $\bra{\phi}    \pi_y    \ket{\phi}    =    \bra{\phi}
  P_{\hilb{K}}  \pi_y  P_{\hilb{K}} \ket{\phi}$,  where
  $P_{\hilb{H}}$ and $P_{\hilb{K}}$  are the projectors
  on   Hilbert   spaces   $\hilb{H}$   and   $\hilb{K}$
  respectively, w.l.o.g.   the first and  second maxima
  in  the  last  step   can  be  taken  over  decodings
  $\{      \pi_y     \in      \lin{\hilb{H}}\}$     and
  $\{ \pi_y \in  \lin{\hilb{K}} \}$ respectively.  Then
  by Def.~\ref{def:utility}  for the first  maximum one
  has
  \begin{align*}
    \max_{\rho_x,   \pi_y}   \sum_{x,y}  \Tr[\rho_x   \pi_y]
    g_{x,y} =: U(\id, g),
  \end{align*}
  and  the second  maximum is  trivially achieved  when
  $\pi_y  = \delta_{y,y^*}  \openone_{\hilb{K}}$, where
  $y^* := \arg\max_y \sum_x g_{x,y}$, namely
  \begin{align*}
    \max_{\pi_y} \sum_{x,y} \bra{\phi}  \pi_y \ket{\phi} g_{x,y}
    = \max_y \sum_x g_{x,y}.
  \end{align*}

  By explicit computation the encodings $\{ \rho_x^* \}$ and
  decodings     $\{     \pi_y^*    \oplus     \delta_{y,y^*}
  \openone_{\hilb{K}} \}$ saturate this upper bound.
\end{proof}

We conclude our study of the utility of quantum channels for
arbitrary games  by considering  quantum-classical channels,
which  are used  to represent  the most  general demolishing
quantum measurement.

\begin{dfn}[Quantum-classical channel]
  The  action  of  the quantum-classical  (QC)  channel
  $\ch{M}: \lin{\hilb{H}} \to \lin{\hilb{K}}$ over
  any state $\rho \in \lin{\hilb{H}}$ is given by
  \begin{align*}
    \ch{M}(\rho) := \sum_y \Tr[\rho \pi_y] \ket{y}\! \! \bra{y},
  \end{align*}
  for some  POVM $\{  \pi_y \in \lin{\hilb{H}}  \}$ and
  some o.n.b. $\{ \ket{y} \in \hilb{K} \}$.
\end{dfn}

\begin{prop}
  \label{thm:qcchannel}
  For any game $g$,  the utility $U(\ch{M}, g)$ of
  QC channel $\ch{M}$ is given by
  \begin{align*}
    U(\ch{M}, g) = \max_S \sum_x \lambda_x,
  \end{align*}
  where $S$ is any deterministic stochastic matrix, that is,
  any  matrix with  entries  $0$ or  $1$  such that  $\sum_z
  S_{z,y} = 1$  for any $y$, and $\lambda_x$  is the largest
  eigenvalue of  $\sum_{y, z}  g_{x,z} S_{z, y}  \pi_y$.  If
  $\ket{\lambda_x}$  is the  corresponding eigenvector,  the
  optimal encoding is given  by $\rho_x = \ket{\lambda_x} \!
  \!  \bra{\lambda_x}$.
\end{prop}

\begin{proof}
  Since  classical decodings  are represented  by stochastic
  matrices, due to the linearity  of the figure of merit the
  optimal   decoding    is   a    deterministic   stochastic
  matrix. Hence, by Def.~\ref{def:utility} one has
  \begin{align*}
    U(\ch{M}, g)  := &  \max_{\{\rho_x\}, S} \sum_{x,  y, z}
    g_{x,z} S_{z,y} \Tr \left[\rho_x  \pi_y \right] \\ \le &
    \max_S \sum_x \lambda_x,
  \end{align*}
  where  the inequality  is  saturated if  and only  if
  encoding   $\{   \rho_x   \}$    is   as   given   in
  Proposition~\ref{thm:qcchannel}.
\end{proof}

\section{Utility for unbiased games}
\label{sec:unbiased}

In this  Section we consider  unbiased games, which are games
where  any  of Bob's  possible  outcomes  generate the  same
average  payoff.    Due  to  Lemma~\ref{lmm:quasilinearity},
without loss of generality we can take such an average to be
zero.

\begin{dfn}[Unbiased game]
  We  call unbiased  game  any game  $g$  such that  $\sum_x
  g_{x,y} = 0$.
\end{dfn}

\begin{dfn}[Depolarizing channel]
  The    action    of    the    depolarizing    channel
  $\ch{D}_\lambda:        \lin{\hilb{H}}       \to
  \lin{\hilb{H}}$         on          any         state
  $\rho \in \lin{\hilb{H}}$ is given by
  \begin{align*}
    \ch{D}_\lambda(\rho) := \lambda \rho + (1-\lambda) \Tr[\rho]
    \frac{\openone}{d},
  \end{align*}
  where $d  := \dim\hilb{H}$ is  the dimension of  Hilbert space
  $\hilb{H}$.
\end{dfn}

\begin{prop}
  \label{thm:depolarizing}
  For    any   unbiased    game   $g$,    the   utility
  $U(\ch{D}_\lambda,   g)$  of   the  depolarizing
  channel $\ch{D}_\lambda$ is given by
  \begin{align*}
    U(\ch{D}_\lambda, g) = \lambda U(\id, g),
  \end{align*}
  The encoding  $\{ \rho_x  \}$ and  decoding $\{  \pi_y \}$
  attaining $U(\id, g)$ also attain $U(\ch{D}_\lambda, g)$.
\end{prop}

\begin{proof}
  By Def.~\ref{def:utility} one immediately has
  \begin{align*}
    & U(\ch{D}_\lambda,  g)  \\ := & \max_{\{\rho_x\}, \{ \pi_y\}}  \sum_{x,y}
    \left[    \lambda    \Tr[    \rho_x   \pi_y]    g_{x,y}    +
    \frac{1-\lambda}d \Tr[\pi_y] g_{x,y} \right].
  \end{align*}
  Since any  POVM is  a decomposition of  the identity,
  namely    $\sum_y    \pi_y     =    \openone$,    and
  $\sum_x g_{x,y} = 0$ for any $y$, one has
  \begin{align*}
    & \sum_{x,y} \left[ \lambda \Tr[ \rho_x \pi_y] g_{x,y} +
      \frac{1-\lambda}d  \Tr[\pi_y] g_{x,y}  \right] \\  = &
    \lambda \sum_{x,y} \Tr[\rho_x \pi_y] g_{x,y},
  \end{align*}
  for   any   encoding    $\{\rho_x\}$   and   decoding
  $\{\pi_y\}$. Then one has
  \begin{align*}
    \max_{\rho_x, \pi_y}  \lambda \sum_{x,y} \Tr\left[\rho_x
      \pi_y\right] g_{x,y} = \lambda U \left(\id, g\right).
  \end{align*}
  for any unitary channel $\ch{U}$.
\end{proof}

\section{Utility for discrimination games}
\label{sec:discrimination}

In this Section we  consider discrimination games, that
is games where the payoff is a diagonal matrix and thus
the numbers  of inputs and  outputs are equal,  that is
$m = n$.

\begin{dfn}[Discrimination game]
  \label{dfn:discrimination}
  We call  discrimination game  any game $g$  such that
  $g_{x,y}  =   \delta_{x,y}  g_x$,  for   some  $g_x$.
\end{dfn}
  
According to Def.~\ref{dfn:discrimination}, the utility
of  any channel  $\ch{C}$ for  discrimination game
$g$ is given by
\begin{align*}
  U(\ch{C}, g)  = \sup_{\{\rho_x\}, \{\pi_x\}} \sum_x  \Tr[ \ch{C}(\rho_x)
  \pi_x] g_x.
\end{align*}

We have shown in Proposition~\ref{thm:frenkel} that w.l.o.g.
the optimization of the utility  for any unitary channel can
be restricted to a finite set of encodings and decodings. In
the following we specify such  a result by deriving a closed
form for the  utility $U(\ch{U}, g)$ of  any unitary channel
$\ch{U}$ for any discrimination  game $g$.  This extends and
simplifies   the   proof   of   a  result   by   Elron   and
Eldar~\cite{EE07}.

\begin{prop}
  For   any  discrimination   game  $g$,   the  utility
  $U(\ch{U},   g)$   of    any   unitary   channel
  $\ch{U}$ is given by
  \begin{align*}
    U(\ch{U},  g) = \sum_{x=0}^{d-1}
    g_x \Theta(g_x),
  \end{align*}
  where           w.l.o.g.            we           take
  $g_0 \ge  g_1 \ge \dots \ge  g_{n-1}$ and $\Theta(x)$
  is the Heaviside step  function.
\end{prop}

\begin{proof}
	For a unitary channel $\ch{U}$, the unitary operator can be absorbed into the encoding and the decoding such that, without loss of generality, we can consider only the case of the identity channel $\ch{U}=\id$.
	
  If   $g_x    \le   0$    for   any   $x$    then   by
  Lemma~\ref{thm:bound}       one        has       that
  $U(\ch{U},  g) \le  0$.   This bound  is attained  by
  encoding $\{  \rho_x \}$  and decoding $\{  \pi_y \}$
  given by
  \begin{align*}
    \rho_x = \begin{cases} \ket{1} \! \! \bra{1},  & \textrm{
          if }  x = 0,  \\ \ket{0} \! \! \bra{0}, & \textrm  { if }  x >
        0, \end{cases}
  \end{align*}
  and
  \begin{align*}
    \pi_x = \begin{cases} \ket{0} \! \! \bra{0}, & \textrm{ if
      } x = 0, \\  \frac{1}{n-1} (\openone - \ket{0} \! \! \bra{0}), &
      \textrm { if } x > 0. \end{cases}
  \end{align*}
  
  Then  let $g_0  > 0$.   For  any $\{  \rho_x \}$  and
  $\{ \pi_x \}$ one has
  \begin{align*}
    \sum_x \Tr[\rho_x \pi_x] g_x  \le \sum_x \Tr[\rho_x \pi'_x]
    g_x
  \end{align*}
  where for any $x \ge 1$ one has
  \begin{align*}
    \pi'_x := \left \{ \begin{matrix}  \pi_x & \textrm{ if } g_x
        \ge  0, \\  0 &  \textrm  { if  } g_x  < 0  \end{matrix}
    \right.
  \end{align*}
  and  $\pi'_0 =  \openone -  \sum_{x \neq  0} \pi'_x$,
  therefore $\pi'_0 \ge \pi_0$.

  Therefore the discrimination game $g'_x$ given by
  \begin{align*}
    g'_x := g_x \Theta(g(x)),
  \end{align*}
  is                      such                     that
  $U(\ch{U}, g')  = U(\ch{U},  g)$.  Moreover
  $U(\ch{U},  g)$  and  $U(\ch{U},  g')$  are
  attained by the same encoding and decoding.
  
  By denoting with $|| \cdot ||_\infty$ the largest singular
  value, By Def.~\ref{def:utility} one has
  \begin{align*}
    U(\ch{U}, g')  & := \max_{\{\rho_x\},  \{\pi_x\}} \sum_x
    \Tr[\rho_x \pi_x] g'_x \\ & = \max_{\{\pi_x\}} \sum_x
    g'_x || \pi_x ||_{\infty},
  \end{align*}
  where  the  second  equality  represents  an  upper  bound
  saturated when  $\rho_x$ is  the projector on  the largest
  eigenvalue of $\pi_x$, for any $x$.
  
  We relax  the condition $\sum_x  \pi_x = \openone$  to the
  weaker  condition $\sum_x  \Tr[\pi_x]  = d$,  where $d  :=
  \dim\hilb{H}$.  Thus we have  $\sum_x || \pi_x ||_{\infty}
  \le d$. Moreover, since $\pi_x  \le \openone$ for any $x$,
  one has $|| \pi_x ||_{\infty} \le 1$ for any $x$. Then one
  clearly has
  \begin{align*}
    \max_{\{\pi_x\}}   \sum_x    g'_x   ||   \pi_x    ||_{\infty}   =
    \sum_{x=0}^{d-1} g'_x,
  \end{align*}
  where  the equality  represents an  upper bound  saturated
  when $\pi_x = \ket{x} \!\! \bra{x}$  for any $x = 0, \dots
  d-1$, so the statement remains proved.
\end{proof}
 
\section{Utility for binary-output games}
\label{sec:binary}

In  this section  we consider  binary-output games,  which are
games  with $n  = 2$  outputs (but  an arbitrary number  $m$ of
inputs).  It  is perhaps  surprising that  any binary-output
game  can be  recast  as a  binary-input-output (binary  for
short) discrimination game, that is  a diagonal game with $m
= n =  2$ inputs and outputs, parametrized by  a single real
parameter, as shown by the following Lemma.

\begin{lmm}
  \label{lmm:binary}
  For   any  binary-output   game   $g$,   there  exists   a
  binary-input-output discrimination game $g'$ such that
  \begin{align*}
  	U(\ch{C}, g) = a\cdot U(\ch{C}, g') + b\;,
  \end{align*}
  for any channel $\ch{C}$. In the above formula one has $g'
  = \diag(g_0,1-g_0)$, where
  \begin{align*}
    a  & := \sum_x  \left| g_{x,  0}  -  g_{x,  1} \right|,\\
    b  & :=  \sum_x  \min_y g_{x,y},\\
    g_0 & := \frac1a \sum_x \left( g_{x,0}  - g_{x,1} \right) \Theta \left( g_{x,0}  -  g_{x,1} \right),
  \end{align*}
  where   $\Theta(x)$  is   the  Heaviside   step  function.
  Moreover,  the   same  encoding  and   decoding  achieving
  $U(\ch{C}, g')$ also achieve $U(\ch{C}, g)$.
\end{lmm}

\begin{proof}
  For   any  binary-output   game   $g$,  consider   another
  binary-output game $\tilde{g}$ defined as
  \begin{align*}
    \tilde{g}_{x,y}  :=  \frac1a  \left(  g_{x,y}  -  \min_z
    g_{x,z} \right),
  \end{align*}
  for any  $x$ and  $y$.  Hence,  for any  $x$ one  has that
  $\tilde{g}_{x, y} \ge 0$ for any $y$, with equality for at
  least      one     value      of     $y$.       Due     to
  Lemma~\ref{lmm:quasilinearity}    one   immediately    has
  $U(\ch{C}, g) =  a\cdot U(\ch{C}, \tilde{g}) +  b$ and the
  encodings  and  decodings  attaining  $U(\ch{C},  g)$  and
  $U(\ch{C}, \tilde{g})$ are the same.
  
  For any encodings $\{ \rho_x  \}$, any decodings $\{ \pi_y
  \}$, and any $x_0$ and $x_1$ such that $\tilde{g}_{x_0, 1}
  = \tilde{g}_{x_1, 1} = 0$,  let without loss of generality
  $\Tr [  \ch{C} ( \rho_{x_1} )  \pi_0 ] \ge \Tr  [ \ch{C} (
    \rho_{x_0}  )  \pi_0  ]$.   Therefore,  replacing  state
  $\rho_{x_0}$ with state $\rho_{x_1}$ increases the average
  payoff.  Hence,  the utility  is attained when  states $\{
  \rho_x \}$  coincide for any $x$  such that $\tilde{g}_{x,
    1}  =   0$,  and  the   same  for  any  $x$   such  that
  $\tilde{g}_{x,  0} =  0$. Hence,  the utility  $U( \ch{C},
  \tilde{g}) = U  ( \ch{C}, g')$, and  the statement remains
  proved.
\end{proof}

Therefore, without  loss of  generality in the  following we
consider binary  discrimination games  $g$, which  are games
with  $m =  2$ inputs  and $n  = 2$  outputs, such  that the
payoff  $g$  is  diagonal   and  constitutes  a  probability
distribution.
  
For  any  such  a  game,  as  an  immediate  consequence  of
Helstrom's theorem~\cite{Hel76} one has that
\begin{align*}
  U(\ch{C}, g) = \max_{\{ \rho_x \}} \frac12 \left( 1 + \| H
  \|_1 \right).
\end{align*}
where  $|| \cdot  ||_1 :=  \Tr [  | \cdot  | ]$  denotes the
$1$-Schatten norm, and
\begin{align*}
  H := g_0 \ch{C}(\rho_0) - (1-g_0) \ch{C}(\rho_1),
\end{align*}
is the Helstrom matrix.

For any  commutativity-preserving channel, that  is any
channel         $\ch{C}$         such         that
$[\ch{C}(\rho), \ch{C}(\sigma)] = 0$ whenever
$[\rho, \sigma]  = 0$, and  any binary discrimination game $g$,  it is
straightforward  that  the  optimal  encoding  is  also
orthogonal,  that  is   $\braket{\phi_0|\phi_1}  =  0$.
However, it is perhaps  surprising that this fact holds
true also for non-commutativity-preserving channels, as
stated by the following Lemma.

\begin{lmm}
  \label{lmm:orthogonal}
  For  any channel  $\ch{C}$ and  any binary  discrimination
  game $g$,  the utility  $U(\ch{C}, g)$  is attained  by an
  orthonormal   encoding   $\{   \phi_x^*   \}$,   that   is
  $\braket{\phi_{x_0}^*|\phi_{x_1}^*} = \delta_{x_0,x_1}$.
\end{lmm}

\begin{proof}
  For any encoding  $\{ \ket{\phi_x} \}$, let  us define the
  matrix $K$ and a spectral decomposition as follows:
  \begin{align*}
    K  :=  g_0 \ket{\phi_0}\!\!\bra{\phi_0}  -  (1-g_0)
    \ket{\phi_1}\!\!\bra{\phi_1}  =:  \sum_k  \lambda_k
    \ket{k} \! \! \bra{k},
  \end{align*}
  and consider  the dephasing channel $\ch{P}_0(  \cdot ) :=
  \sum_k \bra{k} \cdot \ket{k} \ket{k} \!  \!  \bra{k}$ on a
  basis of eigenvectors of $K$.  One has that
  \begin{align*}
    K  =  \ch{P}_0(K)  = g_0  \sigma_0  -  (1-g_0)
    \sigma_1,
  \end{align*}
  where    $\sigma_x     :=    \ch{P}_0(\ket{\phi_x}    \!\!
  \bra{\phi_x})$   and    therefore   $\sigma_x    \ge   0$,
  $\Tr[\sigma_x]  =  1$,  and $[\sigma_0,  \sigma_1]  =  0$,
  namely $\sigma_x$ are  commuting states.  Since $U(\ch{C},
  g)$ only depends  on the encoding  $\{ \phi_x \}$ through
  the  Helstrom  matrix  $H   :=  \ch{C}(K)$,  encoding  $\{
  \sigma_x \}$ performs as well  as encoding $\{ \phi_x \}$,
  and therefore without loss  of generality one can maximize
  over  commuting  encodings  only.   By  the  convexity  of
  $\Tr[\pp(X-Y)]$  in  $X$  and   $Y$,  a  pure  orthonormal
  encoding $\{ \phi_x \}$ suffices.
\end{proof}

Notice that Lemma~\ref{lmm:orthogonal} cannot be generalized
to  discrimination games  with more  than two  alternatives.
For example,  let $\ch{M}$ be the  quantum-classical channel
corresponding  to  the  trine  POVM  of  a  qubit,  that  is
$\ch{M}(\rho)  =   \sum_{y}  \bra{\pi_y}   \rho  \ket{\pi_y}
\ket{y} \!   \!  \bra{y}$  with $\ket{\pi_y} :=  \frac23 U^y
\ket{0}$ and  $U := e^{-i\frac{2 \pi}{3}\sigma_Y}$,  and let
$g$ be  the discrimination  game $g_{x,y}  := \delta_{x,y}$.
Then one has  $U(\ch{M}, g) = 2$ and  the optimal encoding
is the trine encoding, that is $\ket{\phi_x} = U^x \ket{0}$,
which is of course  not pairwise commuting.  For comparison,
the  best pairwise  commuting  encoding  is $\ket{\phi_0}  =
\ket{0}$ and $\ket{\phi_1} = \ket{\phi_2} = \ket{1}$, and in
this  case  the  average  payoff  is  given  by  $\sum_{x,y}
|\braket{\phi_x|\pi_y}|^2 = 5/3 < 2$.

The  Pauli channel  is an  example of  a qubit  channel
which is commutativity preserving.

\begin{dfn}[Pauli channel]
  \label{dfn:pauli}
  The     action      of     the      Pauli     channel
  $\ch{P}_{\vec\lambda}:     \lin{\hilb{H}}    \to
  \lin{\hilb{H}}$  with  $\dim(\hilb{H})  = 2$  on  any
  state $\rho$ is given by
  \begin{align*}
    \ch{P}_{\vec\lambda}(\rho) = \lambda_0 \rho + \sum_{k=1}^3 \lambda_k \sigma_k \rho \sigma_k,
  \end{align*}
  where $\sigma_1  = \sigma_X$, $\sigma_2  = \sigma_Y$,
  and $\sigma_3 = \sigma_Z$ are the Pauli matrices.
\end{dfn}

\begin{prop}
  For    any    binary discrimination   game   $g$,    the    utility
  $U(\ch{P}_{\vec\lambda},  g)$  of the  Pauli  channel
  $\ch{P}_{\vec\lambda}$ is given by
  \begin{align*}
    U(\ch{P}_{\vec\lambda},    g)     =    \max    \left(
    g_0,
    \frac{1+\max_{k \ge 1}|2(\lambda_0+\lambda_k)-1|}2 \right).
  \end{align*}
\end{prop}

\begin{proof}
  By   Lemma~\ref{lmm:orthogonal}   it   suffices   to
  consider     a      pure     orthonormal     encoding
  $\{  \phi_\pm  \}$.  Upon  decomposition  over  Pauli
  matrices one has
  \begin{align*}
    \phi_\pm  = \frac{\openone}2  \pm  \sum_{i  = 1}^3  \alpha_i
    \frac{\sigma_i}2, \qquad \sum_{i=1}^3 \alpha_i^2 = 1.
  \end{align*}
  By direct  computation one has
  \begin{align*}
    \ch{P}_{\vec\lambda}(\phi_\pm)  =  \frac{\openone}2  \pm
    \sum_{i  =}^3 \alpha_i  (2(\lambda_0 +  \lambda_i) -  1)
    \frac{\sigma_i}2
  \end{align*}
  Since the eigenvalues of the  state $\openone + x \sigma_x
  + y  \sigma_y + z \sigma_z$  are $1 \pm \sqrt{x^2  + y^2 +
    z^2}$,  one has  that  the eigenvalues  of the  Helstrom
  matrix  $H =  g_0  \ch{P}_{\vec\lambda}(\phi_+) -  (1-g_0)
  \ch{P}_{\vec\lambda}(\phi_-)$ are
  \begin{align*}
    \frac{2g_0-1     \pm     \sqrt{\sum_{i=1}^3     \alpha_i^2
        [2(\lambda_0 + \lambda_i) -1]^2 }}2,
  \end{align*}
  and thus
  \begin{align*}
    &  U(\ch{P}_{\vec\lambda},  g)  \\ =  &  \max  \left(
                                                 g_0,  \max_{\substack{\vec\alpha \\  \sum_i
    \alpha_i^2   =   1}}  \frac{1   +   \sqrt{\sum_{i=1}^3
    \alpha_i^2   [2(\lambda_0  +   \lambda_i)  -1]^2   }}2
    \right).
  \end{align*}
  By making  the substitution  $\beta_i := \alpha_i^2$,  one has
  that $\vec\beta$  is a probability distribution  and therefore
  the maximum over $\vec\alpha$ can be explicitly computed as
  \begin{align*}
    \max_{\substack{\vec\beta \\  \sum_i \beta_i = 1  \\ \vec\beta
        \ge   0}}\sqrt{\sum_{i=1}^3   \beta_i   [2(\lambda_0   +
      \lambda_i) -1]^2} = \max_{i\ge1} |2(\lambda_0 + \lambda_i)
    - 1|,
  \end{align*}
  and the statement immediately follows.
\end{proof}

The amplitude-damping channel is  an example of a qubit
channel which is not commutativity preserving.

\begin{dfn}[Amplitude damping channel]
  \label{dfn:ampdamp}
  The   action  of   the   amplitude  damping   channel
  $\ch{A}_\eta: \lin{\hilb{H}} \to \lin{\hilb{H}}$
  with    $\dim(\hilb{H})   =    2$   on    any   state
  $\rho := \begin{pmatrix}
    1-\beta & \gamma \\
    \gamma^* & \beta
      \end{pmatrix}$is given by
  \begin{align*}
    \ch{A}_\eta(\rho) = \begin{pmatrix}
      1-\eta \beta & \sqrt{\eta}\gamma \\
      \sqrt{\eta}\gamma^* & \eta \beta
    \end{pmatrix}.
  \end{align*}
\end{dfn}

\begin{prop}
  \label{prop:ampdamp}
  The utility $U(\ch{A}_\eta, g)$ of any amplitude
  damping  channel  $\ch{A}_\eta$ for  any  binary discrimination
  game $g$ is given by
  \begin{align*}
    U(\ch{A}_\eta,  g)  =  \frac{1  +  \sqrt{1  -  4g_0(1-\eta)
    +4g_0^2 (1-\eta)}}2,
  \end{align*}
  and an optimal encoding is given by
  \begin{align*}
    \ket{\psi_0^*}  =  \sqrt{g_0}   \ket{0}  +  \sqrt{1-g_0}
    \ket{1},  \\  \ket{\psi_1^*}  = \sqrt{1-g_0}  \ket{0}  -
    \sqrt{g_0} \ket{1}.
  \end{align*}
\end{prop}

\begin{proof}
  Due to Lemma~\ref{lmm:orthogonal}  the optimal encoding is
  pure and orthonormal. W.l.o.g. we consider
  \begin{align*}
    \ket{\psi_0} = \sqrt{1-\gamma^2} \ket{0} + \gamma \ket{1}, \\
    \ket{\psi_1} = \gamma \ket{0} - \sqrt{1-\gamma^2} \ket{1}.
  \end{align*}
  The  eigenvalues $\lambda_{\pm}$  of the Helstrom  matrix
  $H := g_0 \ch{A}_\eta(\ket{\psi_0} \!\! \bra{\psi_0})
  -      (1-g_0)\ch{A}_\eta(\ket{\psi_1}     \!      \!
  \bra{\psi_1})$ are given by
  \begin{align*}
    \lambda_{\pm}  = \frac{2g_0  - 1  \pm \sqrt{  a\gamma^4 +  b
        \gamma^2 + c }}2,
  \end{align*}
  where
  \begin{align*}
    a = & 4\eta(\eta-1),\\
    b = & 8\eta \left[ (\eta-1)g_0-\eta + 1 \right], \\
    c  =  &  4(\eta-1)^2g_0^2  -4   (2\eta^2  -  3\eta  +  1)g_0
    +4\eta(\eta-1) + 1.
  \end{align*}
  W.l.o.g.,  we   take  $g_0   \ge  \frac12$.    Hence,  the
  discrimination utility $U(\ch{A}_\eta, g)$ is given by
  \begin{align*}
    U(\ch{A}_\eta,  g)  =   \max_{\gamma}  \frac{1+\lambda_+  +
    |\lambda_-|}2.
  \end{align*}

  If  the  utility  is   achieved  for  $\gamma$  such  that
  $\lambda_- \ge 0$, by direct inspection one has
  \begin{align*}
    \frac{1+\lambda_+ + |\lambda_-|}2 = g_0,
  \end{align*}
  for any $\eta$, namely
  \begin{align*}
    U(\ch{A}_\eta, g) = g_0.
  \end{align*}
  However,   this  is   absurd,   because  upon   setting
  $\gamma  = 0$  (i.e.  $\ket{\psi_x}  = \ket{x}$) it immediately follows that
  \begin{align*}
    \frac{1+\lambda_+ + |\lambda_-|}2 = g_0 + (1-g_0)\eta,
  \end{align*}
  namely
  \begin{align*}
    U(\ch{A}_\eta, g) \ge g_0 + (1-g_0)\eta.
  \end{align*}
  Therefore  $U(\ch{A}_\eta,  g)  > g_0$  for  any
  $0 < \eta \le 1$ and $\gamma$ and $g_0 < 1$, which is
  absurd.  Then one can conclude $\lambda_- < 0$ and
  \begin{align*}
    U(\ch{A}_\eta,  g)  =   \max_{\gamma}  \frac{1+\lambda_+  -
    \lambda_-}2.
  \end{align*}
  
  By  direct   inspection,  the  zeros  of   the  first
  derivative  of  $(1+\lambda_+   -  \lambda_-)/2$  are
  attained  when   $\gamma  =  \pm   \sqrt{1-g_0}$  and
  $\gamma = 0$.  The zeros of the second derivative are
  attained when $\gamma = \pm \sqrt{(1-g_0)/3}$ and the
  second  derivative is  positive  when  $\gamma =  0$.
  Therefore    the    maximum    is    attained    when
  $\gamma  =  \pm   \sqrt{1-g_0}$,  and  the  statement
  follows by direct computation.
\end{proof}

Let  us  consider  the  two   extremal  cases  $g_0  =  1/2$
(balanced) and $g_0 = 1$ (maximally unbalanced).  For $g_0 =
1/2$ the optimal  encoding given by Prop.~\ref{prop:ampdamp}
becomes $\ket{\psi_\pm} =  \ket{\pm}$, and the corresponding
binary discrimination  utility becomes $U(\ch{A}_\eta,  g) =
(1 + \sqrt{\eta})/2$.   For $g_0 = 1$,  the optimal encoding
given     by      Proposition~\ref{prop:ampdamp}     becomes
$\ket{\psi_i}  =  \ket{i}$,  and  the  corresponding  binary
discrimination  utility  becomes  $U(\ch{A}_\eta, g)  =  1$.
This situation is depicted in Fig.~\ref{fig:ampdamp}.

\begin{figure}[h]
  \includegraphics[width=\columnwidth]{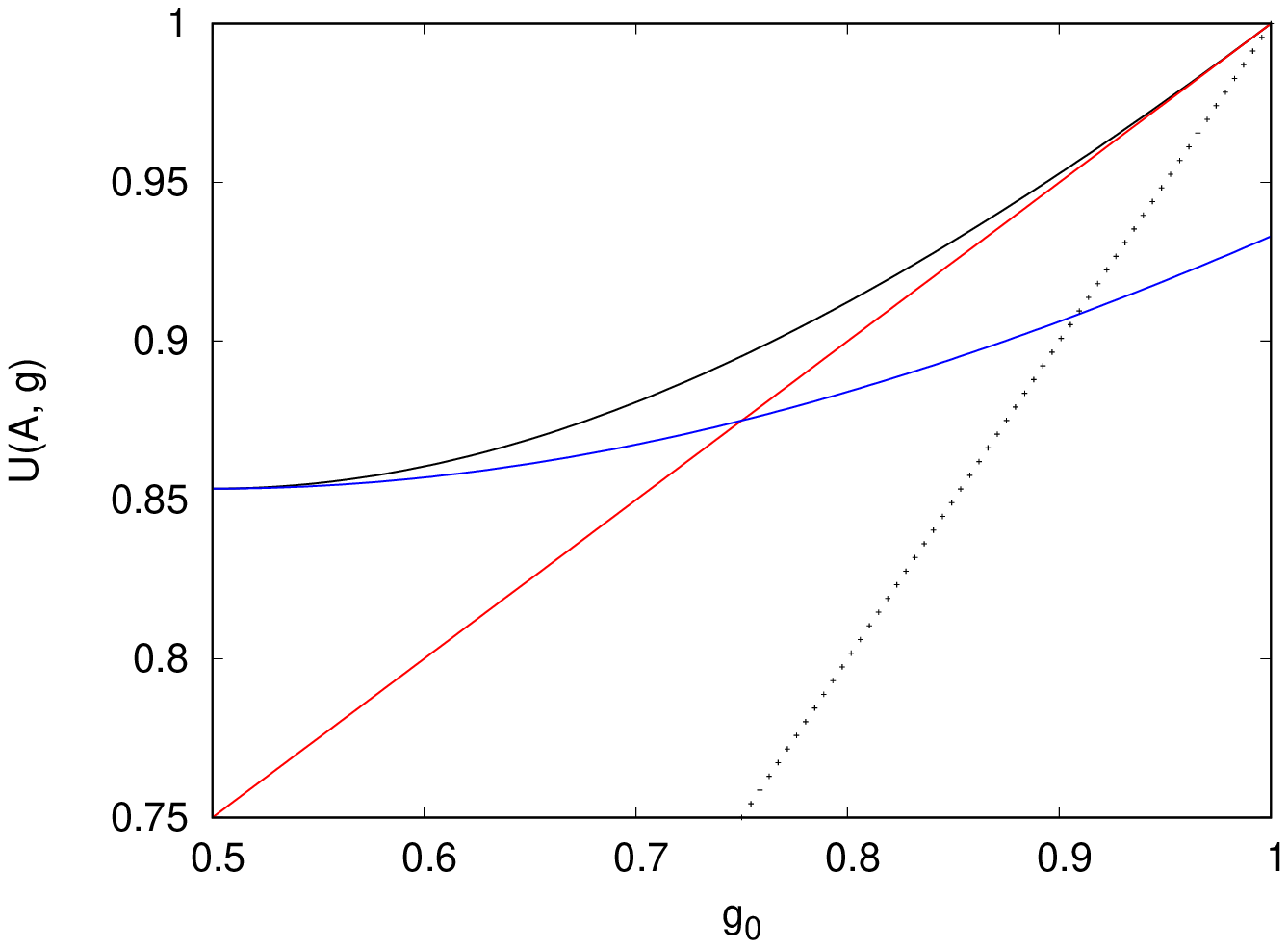}
  \caption{(Color      online)      Binary      utility
    $U(\ch{A}_\eta,  g)$  (upper  black  line)  of  the
    amplitude  damping  channel   $\ch{A}_\eta$  for  a
    binary              discrimination             game
    $g =  \diag(g_0, 1 -  g_0)$ as a function  of $g_0$
    for  fixed $\eta  = \frac12$.   The average  payoff
    with  encodings  $\ket{\psi_i} =  \ket{\pm}$  (blue
    curved  line)  and  $\ket{\psi_i} =  \ket{i}$  (red
    straight line)  are optimal  for the  balanced case
    $g_0 =  \frac12$ and  for the  maximally unbalanced
    case $g_0  = 1$,  respectively.  The  trivial guess
    (dotted  black   line)  is  optimal  only   in  the
    maximally unbalanced case.}
 \label{fig:ampdamp}
\end{figure}

An    example     of    arbitrary     dimensional,    non
commutativity-preserving      channel       is      the
shifted-depolarizing         channel.         Following
Refs.~\cite{GLR05, KNR05, Fuk05,  Ouy14}, we define the
shifted-depolarizing channel as follows.

\begin{dfn}[Shifted-depolarizing channel]
  The  action   of  the   shifted-depolarizing  channel
  $\ch{S}_\lambda:        \lin{\hilb{H}}       \to
  \lin{\hilb{H}}$         on          any         state
  $\rho    \in    \lin{\hilb{H}}$     is    given    by
  \begin{align*}
    \ch{S}_{\lambda}(\rho)  := \lambda
    \rho  +    (1-\lambda)   \Tr[\rho]
    \sigma,
  \end{align*}
  where $\sigma \in \lin{\hilb{H}}$ is a state.
\end{dfn}

\begin{prop}
  \label{prop:shifted}
  For    any    binary    game   $g$,    the    utility
  $U(\ch{S}_\lambda,  g)$  of the  shifted-depolarizing
  channel $\ch{S}_\lambda$ is given by
  \begin{align}
    \label{eq:shifted}
    & U(\ch{S}_\lambda, g) \nonumber\\ = & \max \left[ g_0, \frac{1+\lambda+(1-\lambda)(1-2s_{d-1})(2g_0-1)}2 \right],
  \end{align}
  where  $s_{d-1}$   is  the  smallest   eigenvalue  of
  $\sigma$.
\end{prop}

\begin{proof}
  Due  to Lemma~\ref{lmm:orthogonal}  w.l.o.g. we  take
  the  encoding   $\{  \rho_x   \}$  to  be   pure  and
  orthogonal.  Denote the  eigenvalues of $\sigma$ with
  $\{ s_0, s_1, ..., s_{d-1} \}$ where w.l.o.g. we take
  $s_0 \geq s_1 \geq ....  \geq s_{d-1}$.  The Helstrom
  matrix becomes
  \begin{align*}
    H & = g_0\ch{S}_\lambda(\rho_0) - (1-g_0) \ch{S}_\lambda(\rho_1) \\
      &   =  -\lambda (1-g_0)  \rho_1
        + \lambda  g_0 \rho_0 + (1-\lambda) (2g_0-1) \sigma.
  \end{align*}
  Let  us set  $R  :=  -\lambda (1-g_0)  \rho_1$  and $C  :=
  \lambda g_0 \rho_0 + (1-\lambda) (2g_0-1) \sigma$, so that
  $H = R + C$, where  the only non-null eigenvalue of $R$ is
  $r_{d-1} := -\lambda(1-g_0) \le  0$ and the eigenvalues of
  $C$ are $\{c_0, ...., c_{d-1}  \}$ where $c_j \geq c_{j+1}
  \geq 0$.
  
  Then  we   can  label  the  eigenvalues   of  $H$  as
  $\{h_{0}, ..., h_{d-1}\}$ where $h_{j} \geq h_{j+1}$,
  and applying Weyl's inequality immediately results in
  the system:
  \begin{align*}
    \begin{cases}
      r_{i}+ c_{d-1} \leq h_{i} \leq r_{i} + c_{0}, \\
      r_{d-1} + c_{i} \leq h_{i} \leq r_{0} + c_{i}.
    \end{cases}
  \end{align*}
  Since $r_i = 0$ for any $i < d-1$, this becomes
  \begin{align*}
    \begin{cases}
      c_{d-1} \leq h_{i} \leq c_{i} \ \forall i < d-1, \\
      r_{d-1}  + c_{d-1}  \leq h_{d-1}  \leq r_{d-1}  +
      c_{0}.
    \end{cases}
  \end{align*}
  
  If $h_{d-1}  \geq 0$, then $\|  H \|_1 = 2  g_0 - 1$,
  which   corresponds  to   the  strategy   of  trivial
  guessing.   However,   if  $h_{d-1}  \leq   0$,  then
  $\| H  \|_1 \leq \sum_{i=0}^{d-2} c_{i}  - (r_{d-1} +
  c_{d-1})$.   This  upper  bound can  be  obtained  by
  selecting $R$ to have  its single non-zero eigenvalue
  $r_{d-1}$ corresponding to  the eigenvector for which
  $C$ has the eigenvalue  $c_{d-1}$.  Making use of the
  identity
  $\sum_{i}^{d-2} c_{i} =  \Tr[C]-c_{d-1}= [\lambda g_0
  +  (1-\lambda)   (2g_0-1)  ]  -  c_{d-1}$,   one  has
  $\| H \|_1 = \Tr[C] -  r_{d-1} - 2 c_{d-1}$, and thus
  $\| H \|_1$ is maximized when $c_{d-1}$ is minimized.

  Let us  set $A  = \lambda g_0  \rho_0$ with  only non-null
  eigenvalue $a_0  \ge 0$  and $B =  (1-\lambda) (2g_0  - 1)
  \sigma$ with eigenvalues $\{ b_{0}, ..., b_{d-1} \}$ where
  $b_{j} \geq b_{j+1}$,  so that $C = A +  B$.  Then another
  application  of  Weyl's  inequality yields  $b_{d-1}  \leq
  c_{d-1}$, with equality saturated  if and only if $\rho_0$
  is orthogonal to the eigenvector of $\sigma$ corresponding
  to  eigenvalue $s_{d-1}$  (w.l.o.g. we  take here  $\sigma
  \neq \openone/d$, as the  case of the depolarizing channel
  has        already         been        discussed        in
  Prop.~\ref{thm:depolarizing}),   and   thus   $c_{d-1}   =
  (1-\lambda) (2g_0  - 1) s_{d-1}$.   Finally, $\| H  \|_1 =
  \lambda + (1-\lambda)(1- 2  s_{d-1}) (2g_0-1)$, from which
  the statement immediately follows.
\end{proof}

Further  results  can be  derived  for  the utility  of
group-covariant quantum channels.

\begin{dfn}[Covariant channel]
  A channel $\ch{C}: \lin{\hilb{H}} \to \lin{\hilb{K}}$
  is   $G$-covariant  if   group  $G$   admits  unitary
  representations   $U_k    \in   \lin{\hilb{H}}$   and
  $V_k \in \lin{\hilb{K}}$ such that
  \begin{align}
    \label{eq:covariance}
    \ch{C}(U_k \rho U_k^\dagger) = V_k \ch{C}(\rho) V_k^{\dagger}
  \end{align}
  for any $\rho \in \lin{\hilb{H}}$ and any $k\in G$.
\end{dfn}

\begin{lmm}
  \label{thm:covariance}
  For any $G$-covariant channel $\ch{C}$ and any binary game
  $g$, if  an encoding  $\{ \rho_x  \}$ attains  the utility
  $U(\ch{C}, g)$,  then also  any encodings $\{  \sigma_x :=
  U_k \rho_x  U_k^\dagger \}$, where  $k$ is any  element of
  $G$, attains the same utility $U(\ch{C}, g)$.
\end{lmm}

\begin{proof}
  The statement follows by direct inspection, namely
  \begin{align*}
    & || g_0 \ch{C}(\sigma_0) - (1-g_0) \ch{C}(\sigma_1) ||_1 \\
    = & || g_0 \ch{C}(U_k \rho_0 U_k^\dagger) - (1-g_0) \ch{C}(U_k \rho_1 U_k^\dagger)
        ||_1 \\
    = & || g_0 V_k \ch{C}(\rho_0) V_k^\dagger - (1-g_0) V_k
        \ch{C}(\rho_1) V_k^\dagger ||_1, \\
    = & || V_k \left( g_0 \ch{C}(\rho_0) - (1-g_0)
        \ch{C}(\rho_1) \right) V_k^\dagger ||_1, \\
    = & || g_0 \ch{C}(\rho_0) - (1-g_0) \ch{C}(\rho_1) ||_1,
  \end{align*}
  where      the      second     equality      follows      from
  Eq.~\eqref{eq:covariance}, and the  fourth from the invariance
  of trace distance under unitary transformations.
\end{proof}

Then,  the utility  of  universally covariant  channels
follows. As an example, let  us consider the $1$ to $2$
optimal universally covariant quantum cloning channel~\cite{Wer98}.

\begin{dfn}[Universal optimal cloning]
  The action  of the universal optimal  cloning channel
  $\ch{N}:    \lin{\hilb{H}^{\otimes    N}}    \to
  \lin{\hilb{H}^{\otimes    M}}$     on    any    state
  $\rho \in \lin{\hilb{H}}$ is given by
  \begin{align*}
    \ch{N}(\rho)   :=   \frac{f(N)}{f(M)}  P_s(   \rho^{\otimes N}   \otimes
    \openone^{\otimes (M-N)}) P_s
  \end{align*}
  where $d  = \dim\hilb{H}$, $P_s$ is  the projector on
  the symmetric subspace of $\hilb{H}^{\otimes M}$, and
  $f(x) := \binom{d+x-1}{x}$.
\end{dfn}

\begin{prop}
  \label{thm:cloning}
  For  any  binary  discrimination  game  $g$,  the  utility
  $U(\ch{N}, g)$ of the $1  \to 2$ universal optimal cloning
  channel $\ch{N}: \lin{\hilb{H}} \to \lin{\hilb{H}^{\otimes
      2}}$ is given by
  \begin{align*}
    U(\ch{N}, g) = \frac{d+g_0}{d+1},
  \end{align*}
  and any orthonormal pure encoding is optimal.
\end{prop}

\begin{proof}
  By                        Lemmas~\ref{lmm:orthogonal}
  and~\ref{thm:covariance}    any   orthonormal    pure
  encoding  is   optimal.   W.l.o.g.   let  us   fix  a
  computational          basis         and          set
  $\rho_x = \ket{x} \! \!  \bra{x}$ for $x = 0,1$. Then
  we have
  \begin{align*}
    P_s    =     \sum_{n,m}    \frac{    (\ket{n,m}+\ket{m,n})
      (\bra{n,m}+\bra{m,n}) } 4,
  \end{align*}
  and thus
  \begin{align*}
    \ch{N}(\ket{0} \! \! \bra{0})       =      \frac2{d+1}       \sum_i
    \frac{(\ket{0,i}+\ket{i,0})(\bra{0,i}+\bra{i,0})}4,
  \end{align*}
  and                  analogously                  for
  $\ch{N}(\ket{1} \!  \! \bra{1})$  upon replacing
  $\ket{0}$ with $\ket{1}$. Therefore we have
  \begin{align*}
    &  \ch{N}(\ket{0} \! \! \bra{0})  \ch{N}(\ket{1} \! \! \bra{1})   \\  =  &
                                                                                         \frac1{(d+1)^2}
                                                                                         \frac{(\ket{0,1}+\ket{1,0})(\bra{0,1}+\bra{1,0})}2  \\  =  &
                                                                                                                                                      \ch{N}(\ket{1} \! \! \bra{1}) \ch{N}(\ket{0} \! \! \bra{0}),
  \end{align*}
  and                                              thus
  $[\ch{N}(\ket{0} \! \! \bra{0}), \ch{N}(\ket{1} \! \!
  \bra{1})]          =          0$.           Therefore
  $\ch{N}(\ket{0}\!\!\bra{0})$                      and
  $\ch{N}(\ket{1}\!\!\bra{1})$ admit a  basis of common
  eigenvectors, namely
  \begin{align*}
    \ch{N}(\ket{x}\!\!\bra{x})      =      \sum_k      \lambda_{k|x}
    \ket{k}\!\!\bra{k}.
  \end{align*}
  The only  common eigenvector such that  $\lambda_{k|x} \neq 0$
  for           any           $x$          is           $
  (\ket{0,1}+\ket{1,0})(\bra{0,1}+\bra{1,0})/2$       and      its
  corresponding eigenvalue is $(d+1)^{-1}$.  Then the trace norm
  of the Helstrom matrix is given by
  \begin{align*}
    \|      g_0       \ch{N}(\ket{0}\!\!\bra{0})      -      (1-g_0)
    \ch{N}(\ket{1}\!\!\bra{1})\|_1 = \frac{d + 2g_0 - 1}{d+1},
  \end{align*}
  from which the statement immediately follows.
\end{proof}

Notice that the partial trace  of the $1 \to 2$ cloning
channel $\ch{N}$ is a depolarizing channel given by
\begin{align*}
  \Tr_2[\ch{N}(\rho)]  =  \frac1{2(d+1)}  \left(  (d+2)  \rho  +
  \Tr[\rho] \frac{\openone}d \right) = \ch{D}_\lambda(\rho),
\end{align*}
with  $\lambda =  (d+2)/[2(d+1)]$.  Its  utility is
given by Eq.~\eqref{eq:shifted}, i.e.
\begin{align*}
  U(\Tr_2[\ch{N}(\rho)],     g)    =     \max    \left(     g_0,
    \frac{(d-2)g_0+d+3}{2(d+1)} \right).
\end{align*}
This situation is depicted in Fig.~\ref{fig:cloning}.

\begin{figure}[h]
  \includegraphics[width=\columnwidth]{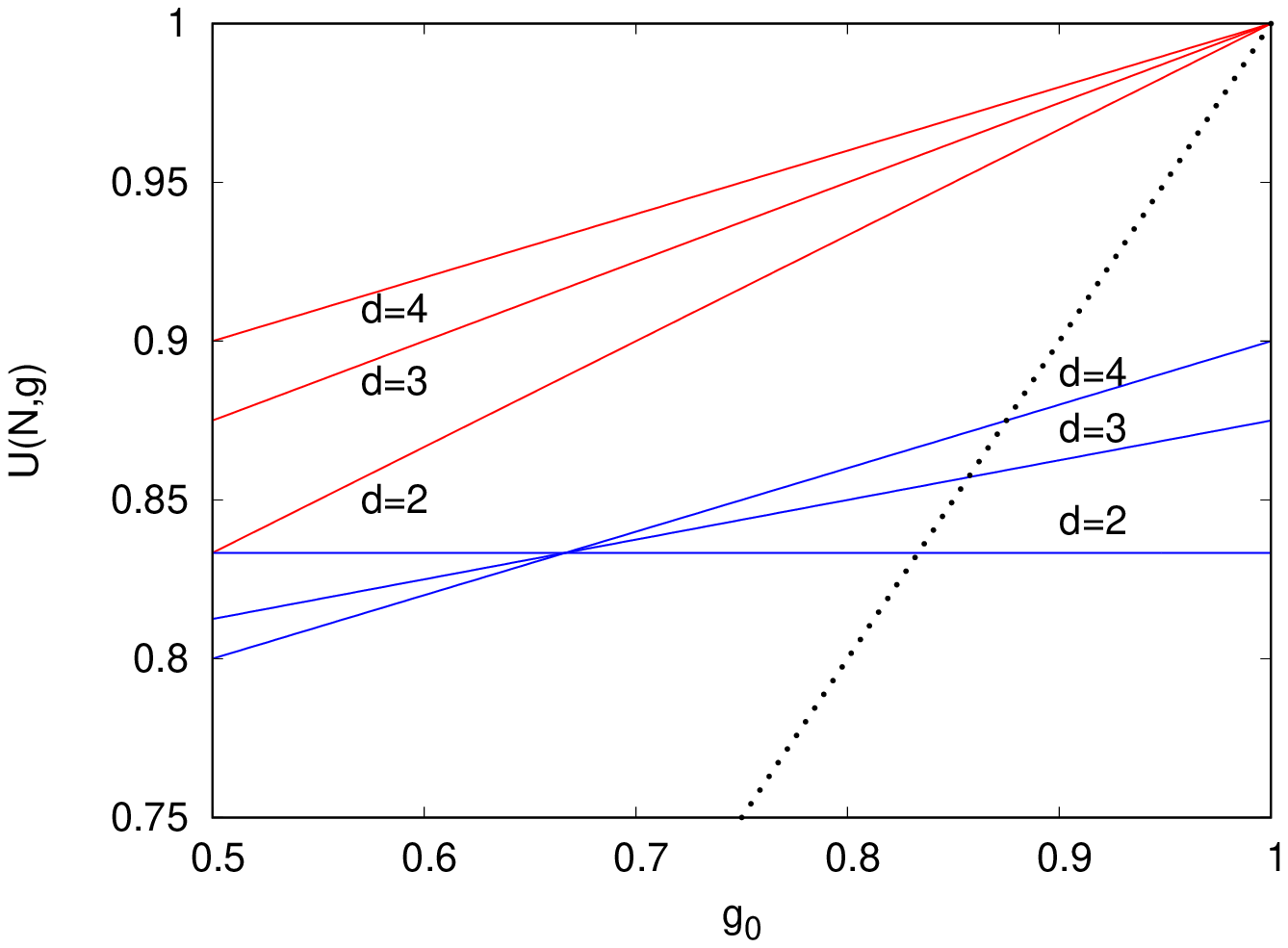}
  \caption{(Color  online)  Binary  utility  $U(\ch{N},  g)$
    (upper red lines)  and binary utility $U(\ch{D}_\lambda,
    g)$ (lower blue lines)  of the covariant cloning channel
    $\ch{N}$ and  depolarizing channel $\ch{D}_\lambda(\rho)
    := \Tr_2[\ch{N}(\rho)]$ with $\lambda = (d+2)/[2(d+1)]$,
    respectively,  as   a  function  of  $g_0$,   for  fixed
    dimension  $d =  2$, $3$,  and $4$.   The trivial  guess
    (dotted line)  is also  depicted.  Notice  that, perhaps
    surprisingly,  for   $d=2$  and  $g_0  =   \frac12$  the
    utilities  $U(\ch{N},  g)$  and  $U(\ch{D}_\lambda,  g)$
    coincide.  This can be  regarded as an entangled analogy
    of the readily verifiable  fact that, for any $\lambda$,
    the  success  probabilities  in  the  discrimination  of
    equiprobable                 qubit                states
    $\ch{D}_\lambda(\ket{0}\!\!\bra{0})$                 and
    $\ch{D}_\lambda(\ket{1}\!\!\bra{1})$  are the  same when
    one or two copies of the unknown state are available.  }
  \label{fig:cloning}
\end{figure}

\section{Conclusion and outlook}
\label{sec:conclusion}

In  this work  we  considered  quantum communication  games,
where  Alice's  task  is  to  communicate  some  information
received by a  referee to Bob through a  quantum channel, in
order to maximize a payoff that depends on both the received
information and  Bob's output.   The maximum  average payoff
defines the  communication utility  of the channel  for that
particular game. Hence, communication games act as witnesses
on the set of classical  noisy channels that can be obtained
from  the  given  quantum  channel,  and  the  corresponding
communication utility constitutes the  optimal value for any
such a witness. We  derived general results and closed-form,
analytic  solutions for  the utility  of several  classes of
quantum channels and several classes of games.

A natural extension of our  setup consists of allowing Alice
and  Bob to  share some  entangled state,  thus generalizing
superdense  coding~\cite{BW92}  to   cases  involving  noisy
communication  channels. In  such a  case, the  object being
considered  is  not  the  channel  $\ch{C}$  alone,  but  an
extension  $\ch{C}\otimes\id$.  For  the noiseless  channel,
clearly  the  entangled  assisted utility  of  the  identity
channel in dimension $d$ is equivalent to the utility of the
identity channel in dimension $d^2$.  Based on this example,
it is  clear that  in at least  some cases  entanglement can
increase the utility of quantum channels.

Our  results  shed new  light  on  the problem  of  creating
quantum  correlations.  It  is a  known result~\cite{HFZL12,
  GC13} that classically correlated  bipartite states can be
transformed by a local channel in a state exhibiting quantum
correlations,   if  and   only   if  the   channel  is   not
commutativity preserving.  However, not  much is known about
the  problem  of   characterizing  commutativity  preserving
channels in  the general  case where  the dimensions  of the
input and  output Hilbert  spaces differ.  Notice  that from
the proof of  Thm.~\ref{thm:cloning}, it immediately follows
that universal optimal $1$ to $2$ cloning is a commutativity
preserving channel.

Our   results    have   important   applications    in   the
quantification of non-markovianity  of quantum channels.  In
Theorem~(3)  of Ref.~\cite{WSB15},  it  was  shown that  the
supremum  of  non-Markovianity  over the  encoding  and  the
binary  discrimination game  is  attained  by an  orthogonal
encoding.  From the proof of our Lemma~\ref{lmm:orthogonal},
it immediately  follows that this  is actually the  case for
{\em any} binary discrimination game.

Finally,  our results  generalize  to the  quantum case  the
partial  ordering   among  classical  channels   derived  by
Shannon~\cite{Sha58}.   Therein, it  is  shown  that if  one
classical channel  $\ch{C}_0$ can  be reproduced  by another
$\ch{C}_1$  upon classical  (possibly  correlated) pre-  and
post-processes, then $\ch{C}_1$  has larger Shannon capacity
than  $\ch{C}_0$.  However,  the  validity  of the  converse
remains an open problem.

\end{document}